\documentclass[letterpaper, 10 pt, conference]{ieeeconf} 
\IEEEoverridecommandlockouts
\usepackage{amsmath}
 
\usepackage{amsthm}
\usepackage{accents}
\usepackage{mathtools}
\usepackage{amssymb}
\usepackage{arydshln}  

\theoremstyle{definition}
\newtheorem{theorem}{\bf Theorem}
\newtheorem{lemma}{\bf Lemma}
\newtheorem{proposition}{\bf Proposition}

\newtheorem{corollary}{\bf Corollary}

\newtheorem{assumption}{\bf Assumption}

\usepackage{graphicx} 
\usepackage[most]{tcolorbox}
\usepackage{xcolor}
\usepackage{todonotes}
\usepackage{comment}
\usepackage[framemethod=TikZ]{mdframed}

\let\labelindent\relax 
\usepackage{enumitem}

\usepackage[noadjust]{cite}  

\newcommand{\q}{\ensuremath{\mathrm{q}}}

\providecommand{\eu}{\ensuremath{\mathrm{e}}}

\newcommand{\defeq}{\vcentcolon=}

\newcommand{\G}{G(\eu^{j\omega_\ell})}

\newcommand{\Gdir}{\widehat{G}_\text{dir}}
\newcommand{\Gind}{\widehat{G}_\text{ind}}
\newcommand{\Gio}{\widehat{G}_\text{io}}
\newcommand{\Tyr}{\widehat{T}_{yr}}
\newcommand{\Tur}{\widehat{T}_{ur}}

\newcommand{\E}{\mathbb{E}}

\makeatletter
\providecommand*{\diff}%
{\@ifnextchar^{\DIfF}{\DIfF^{}}}
\def\DIfF^#1{%
\mathop{\mathrm{\mathstrut d}}%
\nolimits^{#1}\gobblespace}
\def\gobblespace{%
\futurelet\diffarg\opspace}
\def\opspace{%
\let\DiffSpace\!%
\ifx\diffarg(%
\let\DiffSpace\relax
\else
\ifx\diffarg[%
\let\DiffSpace\relax
\else
\ifx\diffarg\{%
\let\DiffSpace\relax
\fi\fi\fi\DiffSpace}

\providecommand*{\deriv}[3][]{%
\frac{\diff^{#1}#2}{\diff #3^{#1}}}

\newlength{\dhatheight}
\newcommand{\doublehat}[2][1.5mu]{%
    \settoheight{\dhatheight}{\ensuremath{\mskip-#1\widehat{#2}\mskip#1}}%
    \addtolength{\dhatheight}{-0.35ex}%
    \mskip#1\widehat{\vphantom{\rule{1pt}{\dhatheight}}%
    \smash{\mskip-#1\widehat{#2}\mskip#1}}\mskip-#1\relax}

\newcommand{\Giodh}{\doublehat{G}_\text{io}}

\newcommand{\asymvar}{\overset{\leadsto}{\text{{var}}}}

\makeatletter
\newcommand*\bigcdot{\mathpalette\bigcdot@{.8}}
\newcommand*\bigcdot@[2]{{\mathbin{\vcenter{\hbox{\scalebox{#2}{$\m@th#1\bullet$}}}}}}
\makeatother

\title{\LARGE \bf Small Noise Analysis of Non-Parametric\\ Closed-Loop Identification }

\author{Mohamed Abdalmoaty and Roy S. Smith
\thanks{This work has been supported by the Swiss National Science Foundation under NCCR Automation (grant agreement  $51\text{NF}40\_180545$)}
\thanks{M. Abdalmoaty and R. Smith are with the Automatic Control Laboratory (IfA) and NCCR Automation of the Swiss Federal Institute of Technology (ETH Z\"{u}rich), 8092 Z\"{u}rich, Switzerland,
        {\tt\footnotesize \{mabdalmoaty,rsmith\}@control.ee.ethz.ch}}
}

\begin{document}
\bstctlcite{IEEEexample:BSTcontrol} 
\maketitle
\thispagestyle{empty}
\pagestyle{empty}

\begin{abstract}
We revisit the problem of non-parametric closed-loop identification in frequency domain; we give a brief survey of the literature and provide a small noise analysis of  the direct, indirect, and joint input-output methods when two independent experiments with identical excitation are used. The analysis is asymptotic in the noise variance (i.e., as the standard deviation of the innovations $\sigma \to 0$), for a finite data record of length $N$. We highlight the relationship between the estimators accuracy and the loop shape via asymptotic variance expressions given in terms of the sensitivity function. The results are illustrated using a numerical simulation example.
\end{abstract}
\section{Introduction}\label{sec:introduction}

Linear closed-loop identification problems have received a significant amount of attention in the Systems and Control community. This is mainly due to the close interaction between system identification and control design; see e.g., \cite{Gevers1993,Vandenhof1995,Zang1995,Hjalmarsson1996} and the references therein. Closed-loop identification is particularly crucial in cases where the open-loop system exhibits instability or when open-loop operation is practically prohibited due to safety or economic considerations.

One of the central problems of closed-loop identification is that  the feedback loop introduces correlations between the measured input and output signals via the unmeasurable disturbances and measurement noise. This makes it difficult to isolate the effect of noise, and hinders accurate identification. To tackle this challenge, several methodologies were developed, in the time-domain as well as the frequency-domain, under varying assumptions on the external excitation signal and the available prior knowledge of the system, if any. They can be broadly divided into ``direct", ``indirect" and ``joint input-output" methods,  see \cite{Ljung1999system,söderström1989system,pintelon2012system}.
Specifically, frequency-domain methods prove beneficial when the  excitation signal can be designed a priori, offering great flexibility for non-parametric modelling where the goal is to estimate the system at a fixed number of frequencies. When the excitation is periodic over time, and the output is measured in steady-state under open-loop operation, the Empirical Transfer Function Estimate (ETFE)  provides unbiased asymptotically optimal estimates as the number of periods $M\to\infty$ \cite[Chap. 6]{Ljung1999system}. When the excitation contains integer number of periods $M$, the accuracy can be improved by arithmetic averaging. The situation gets more complicated when data is collected in closed-loop.

A significant part of the literature on closed-loop identification focused on the use of the prediction-error framework for parametric modelling (see \cite{Forssell1999} and the references therein). Due to the intractability of finite-sample analysis of general parametric models, the performance of these methods is evaluated asymptotically as the data size increases via the first two moments in the asymptotic \cite{Gevers1997} and finite model order cases \cite{Ninnes2005a, Ninnes2005b}. Unfortunately, parametric models come with the complications of under-modelling, which can be alleviated by considering non-parametric models.

The classical papers \cite{Wellstead1977,Wellstead1981} addressed the problem of closed-loop non-parametric identification in frequency domain using the asymptotically unbiased spectral estimator (also known as the joint input-output approach \cite[Chap. 13]{Ljung1999system}; originally introduced in \cite{Akaike1968}) making use of a know exogenous reference signal. First order approximations of the finite-sample bias and variance of the estimator are given; 
however, these approximations should be treated with care, as the theoretical variance of the estimator is generally  infinite \cite{Heath2001_ind_cl_bias}. This is regardless of the fact that the estimator is consistent and asymptotically normal as the data record grows unbounded.

In fact, some open-loop identification methods are plagued by this infinite finite-sample variance problem whenever the excitation signal is stochastic. This is made clear in \cite{Broersen1995}, which compares the finite-sample bias and variance of three open-loop non-parametric estimators (ETFE, spectral estimator, windowed spectral estimator), when the input excitation is a known realization of a Gaussian white stochastic process. 

A finite-sample analysis of the statistical properties of the indirect non-parametric closed-loop  estimator is given in \cite{Heath2000, Heath2001_cl_pdf}.  This analysis involved deriving the finite-sample distribution of the estimator assuming steady-state periodic excitation, and the frequency domain noise being complex Gaussian. However, the latter assumption is unlikely to be satisfied for short data records, but would be valid, under weak conditions, asymptotically in the number of samples,  due to the properties of Fourier transforms (see \cite[Chap. 4]{Billinger2001}). A characterization of the probability density function in terms of the inverse of the controller was given in \cite{Heath2002_cl_var}.  

Although the finite-sample variance 
is infinite, the authors of \cite{Pintelon2001} used an asymptotic argument, in the number of periods  (i.e. as $M \to \infty$), and derived an approximate explicit bias and variance expressions in terms of the Signal-to-Noise (SNR) ratio. They also derived confidence bounds using the approximate variance assuming that the frequency domain noise has a circularly-symmetric complex Gaussian distribution, and relying on a central limit theorem argument.

Using similar analysis and arguments as \cite{Heath2001_cl_pdf}, \cite{Pintelon2003} derived the finite-sample distribution of the estimators as well as uncertainty bounds under the same Gaussian assumption on the noise. Nevertheless, the  expressions  obtained are exact only if the leakage effects in the Discrete Fourier Transform (DFT) are ignored, and the Gaussian assumption is substantiated.  Moreover, the argument used to justify the approximations is again based on a central limit theorem (obtained as $M\to \infty$, a condition that underlies the high SNR assumption there). 
On the other hand, \cite{Welsh2002} proposes a modification of the indirect estimator that is meant to remove its singularity, leading to a finite  variance (also see \cite{Goodwin2001}).
\smallskip

In this contribution we revisit the problem of non-parametric closed-loop identification in frequency domain. We take a different approach and provide a small noise analysis of  the direct, indirect, and joint input-output methods when an independent measurement of the plant input under identical periodic excitation is available. The analysis is asymptotic in the noise variance (i.e., as the standard deviation of the innovations $\sigma \to 0$), for a finite data record of length $N$; we call this a small noise analysis. It does not rely on any distributional assumptions, nor any central limit theorems. We highlight the relationship between the accuracy of the estimators and the loop shape via asymptotic variances expressions given in terms of the sensitivity function. The results are given for the single-input single-output case, but we note that they can be extended to cover the case with multiple inputs and outputs. 

\subsection*{Notation}
The symbol  $\ast$ denotes complex conjugate if it appears as a superscript; otherwise, it denotes a convolution, $\Re[x]$ denotes the real part of a complex random variable $x$. The symbol $\leadsto$ denotes convergence in distribution (weak convergence), and $\asymvar$ denotes the variance of the asymptotic distribution as $\sigma\to 0$.
\section{Experimental Configuration}\label{sec:formulation}
We will consider the closed-loop configuration shown in Figure \ref{figure:cl}, where $C$ and $G$ are discrete-time controller and plant respectively, $r_1$ and $r_2$ are two known exogenous signals, $u$ is the plant input, $y$ is the plant output. The signal $v$ is a discrete-time stationary stochastic process with a spectral factor (noise model) $H$; it represents output measurement noise as well as any unmeasurable disturbances entering the loop, while $e$ is a white innovation process with zero mean and standard deviation $\sigma$. We will simply refer to $v$ as ``noise".  Without loss of generality, we will assume a unit sample time, and for simplicity we define
\begin{equation*}
    r(k) \defeq r_2(k) + C(\q) r_1(k)
\end{equation*}
 where  $k = 0, 1, 2, \dots$ denotes the sample number, and assume that one of the exogenous signals may be used as a probing signal. The closed-loop system equations are
\begin{equation*}
\begin{aligned}
    y(k) &= S(\q) G(\q) r(k) + v_y(k)\\
    u(k) &= \hphantom{G(\q)}S(\q)r(k) - v_u(k)\\
\end{aligned}
\end{equation*}
where $\q$ denotes the forward shift operator,
\[
v_y(k) \defeq S(\q) v(k), \qquad v_u(k) \defeq S(\q) C(\q) v(k),
\]
and 
\[
S(\q) \defeq \frac{1}{1 + G(\q)C(\q)}
\]
is the sensitivity function.

\begin{figure}
\centering
\includegraphics[width=0.7\linewidth]{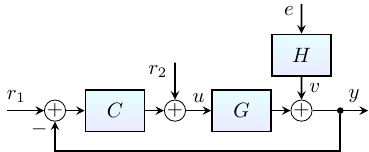}
\caption{Closed-loop configuration}
\label{figure:cl}
\end{figure}

An identification experiment will consist of providing an excitation signal $r$ and recording the corresponding input and output of the plant, resulting in a data set of reference-input-output triplets 
\[
D_N \defeq \{(r_k, u_k , y_k) \; | \; k = 0, \dots, N-1\}.
\]

We define the $N$-point DFT of the output $\{y_k\}_{k=0}^{N-1}$, as 
\begin{align*}
Y(\eu^{j\omega_\ell}) &\defeq \frac{1}{\sqrt{N}}\sum_{k=0}^{N-1} y_k \eu^{-j\omega_\ell k}, \\
\omega_\ell &= \frac{2\pi \ell}{N}, \quad \ell = 0. \dots, N-1,
\end{align*}
and similarly for the input, reference, and noises $v$, $v_y$, $v_u$.
All estimators and analysis will be given in the frequency domain, and since the estimators are defined for each frequency independently, we will drop the dependence of frequency from all the notations. 
\smallskip

We will adapt the following standard assumptions.

\begin{assumption}\label{assumption}
\hspace{0px}
\begin{enumerate}[label={\roman*)}]
    \item $G$ and $C$ are linear time-invariant, and the closed-loop system is stable,
    \item $r$ is periodic with a period $N$,
    \item $D_N$ is collected in steady-state, consists of one period of $r$, and  all signal records are time-synchronized with respect to $r$,
\end{enumerate}
\end{assumption}

The ultimate goal is to estimate the plant $\G$ at a uniform grid of frequencies $\omega_0, \dots, \omega_{N-1}$ using non-parametric estimators: $D_N \mapsto \widehat{G}(\eu^{j\omega_\ell})$. An integral part of the problem is the analysis of the properties of different estimation methods.

\section{Non-parametric Identification Methods}\label{sec:cl_methods}
When Assumption \ref{assumption} is satisfied, the deterministic parts of $u$ and $y$ are periodic and contain no transient. Their DFTs are related as
\begin{align*}
Y &= S G R + V_y\\
U &= \quad S R + V_u
\end{align*}
where $V_y $ and $V_u$ are the noise contributions to the DFT of the output and input, respectively. 

Define the following two (open-loop) estimators
\begin{equation}\label{eq:cl-tfs}
\begin{aligned}
     \Tyr &\defeq \frac{Y}{R} =&      SG \;+\; \frac{V_y}{R}, \\
     \Tur &\defeq \frac{U}{R} =& \quad S \;+\; \frac{V_u}{R}.
     \end{aligned}
\end{equation}
These are ETFEs of the closed-loop transfer functions $T_{yr} = SG$ from $r$ to $y$ and $T_{ur} = S$ from $r$ to $u$, respectively. They enjoy favored statistical properties for periodic excitation. Furthermore, they make minimal assumptions on the underlying system.
\smallskip

Now let us consider the following three closed-loop methods (\cite{Forssell1999,Heath2002_cl_var}):\medskip
\begin{enumerate}[label=\emph{\Alph*.}, itemsep= 2em]
\item \textit{The direct method}
\begin{equation}\label{eq:dir}
    \Gdir \defeq \frac{Y}{U} = \frac{G + {V_y}{(SR)^{-1}}}{\; 1+ {V_u}{(SR)^{-1}}} 
\end{equation}

\item \emph{The indirect method}
\begin{equation}\label{eq:ind}
    \Gind := \frac{\Tyr}{1-\Tyr C} = \frac{G + {V_y}{(S R)^{-1}}}{ \quad 1 - {V_y C }{(S R)^{-1}}}
\end{equation}

\item \emph{The joint input-output method}
\begin{equation}\label{eq:io}
    \Gio:= \frac{\Tyr}{\Tur}
\end{equation}
\end{enumerate}

For all three methods, the estimator of $G$ is constructed as a ratio of two estimators. The numerator is an estimator of $T_{yr}$, and the denominator is an estimator of the sensitivity function. It is clear that in the noise-free case, all estimators coincide with $G$. When noise is present, and depending on its distribution, there may be  a non-zero probability of division by zero, in particular at low frequencies where $S$ is usually small. Due to this singularity issue, the finite-sample variance of these estimators may be infinite.

We now observe the following. Firstly, due to the use of periodic excitation, $\Gdir = \Gio$. 
Secondly, if the leakage effects in the noise DFTs are neglected (so that $V_y = SV$, and $V_u = SCV$), the noise term in the denominator of $\Gind$ coincides with that of $\Gdir$ except for its sign.

When employing just one experiment, the numerator and denominator of each of these three estimators are statistically dependent. Questions of interest arise: How does the accuracy of estimation change when the numerator and denominator become statistically independent by construction; that is, by utilizing two independent estimates of $T_{yr}$ and $S$? Does the statistical dependence between the numerator and denominator improve or worsen the accuracy?

\section{Small noise analysis}\label{sec:analysis}
Let us consider the following two cases.

\begin{enumerate}[label= , itemsep =1em, ]
    \item {\bf{\textit{Case 1:}}} Only one experiment is available. In this case, the estimators are defined as in the previous section, and $\Gdir = \Gio$. Essentially, the direct and indirect estimators are also equivalent up to leakage effects in the  DFT of the noises. This difference appears because noise is not periodic and is filtered differently by the two estimators; it disappears asymptotically in the total number of samples $N$ with a rate $\mathcal{O}_p(N^{-\frac{1}{2}})$ when $H$ is rational/strictly stable (see \cite[Sec. 6.3]{Ljung1999system}).

    \item {\bf{\textit{Case 2:}}} Here we consider the possibility of generating two independent estimates of the closed-loop transfer functions \eqref{eq:cl-tfs}. For simplicity, assume that the same excitation signal $r$ is used, and an independent measurement of $u$ is possible. This means that the deterministic parts of the measured signals are identical, but the noises are independent. Let us denote the DFTs of the output as $Y^{(1)}$, and the input as $U^{(2)}$ where the subscript denotes the experiment used to generate the estimates, and similarly denote the DFTs of noises.
    Now let 
    \[
    \hat{T}_{yr}^{(1)} = \frac{Y^{(1)}}{R}, \quad   \hat{T}_{ur}^{(2)} = \frac{U^{(2)}}{R},
    \]
     and define the following joint input-output estimator:
    \begin{align}
        \Giodh = \frac{\hat{T}_{yr}^{(1)}}{\hat{T}_{ur}^{(2)}} =\frac{G + {V_y^{(1)}}{(SR)^{-1}}}{\; 1+ {V_u^{(2)}}{(SR)^{-1}}} \label{eq:dir2}
    \end{align}
    This estimator is different from \eqref{eq:dir} and \eqref{eq:io} in that the numerator and denominator of the defining fraction are independent by construction. We use the double-hat superscript to denote this important difference.
\end{enumerate}

We are interested in the performance of the estimators constructed in these two cases.  Before studying this, we have the following proposition.

\begin{proposition}\label{prop:finite_dft_statistics}
 It is possible to write
\[
V_y = \sigma \bar{V}_y, \qquad V_u = \sigma \bar{V}_u
\]
where $\bar{V}_y$ and $\bar{V}_u$ are complex random variables with zero mean and finite variances. Furthermore,
\begin{equation}\label{eq:variances}
\begin{aligned}
    \bar{\sigma}_{y} &= \E[\bar{V}_y \bar{V}_y^\ast] = F_N \ast |SH|^2  , \\ 
    \bar{\sigma}_{u}&=\E[\bar{V}_u \bar{V}_u^\ast] = F_N \ast |SCH|^2  \\
    \bar{\sigma}_{yu} &= \E[\bar{V}_y \bar{V}_u^\ast] = F_N \ast   |SH|^2 C^*
\end{aligned}
\end{equation}
where $F_N$ is the Fej\'{e}r kernel on the circle (\cite{Stein2011fourier}).

\end{proposition}
\begin{proof}
    The first part follows directly by the linearity of the DFT and the stability of the closed-loop system. Indeed, one can write $v_y(k) = \sigma\, S(\q)\bar{v}(k)$, where $\bar{v}(k)$ is a normalized version of $v$ with variance $\frac{1}{\pi}\int_{-\pi}^\pi|H(\eu^{j\omega})|^2 \diff \omega$. The same argument is used for $v_u(k)$.
    The second part follows from the definition of the DFT and the linearity of the expectation operator. Let $\rho(k)$ be the correlation function of ${v_y}$. Then, the variance of $\bar{V}_y$ at frequency $\ell$ is
    \[
        \mathbb{E}[\bar{V}_y \bar{V}_y^\ast] = \frac{1}{N} \sum_{m=0}^{N-1}\sum_{n=0}^{N-1} \mathbb{E}[v_m v_n] \eu^{-j \omega_\ell m} \eu^{j \omega_\ell n},
    \]
    and by stationarity, and letting $k = m-n$,
    \[
    \begin{aligned}
        &= \frac{1}{N} \sum_{n=0}^{N-1}  \sum_{k=-n}^{N-1-n} \rho(k) \eu^{-j \omega_\ell k}  \\
        &= \sum_{n=-(N-1)}^{N-1} f_N(n) \rho(n) \eu^{-j \omega_\ell n} 
    \end{aligned}
    \]
    where  $f_N(n) = (N-|n|)/N$ for $|n|\leq N$ and zero otherwise. This means we get a convolution with the Fej\'{e}r kernel $F_N(\omega) := \sum_{k = -\infty}^{+\infty} f_N(k) \eu^{-j\omega k}$ in the frequency domain:
    \[
    \begin{aligned}
   \mathbb{E}[\bar{V}_y \bar{V}_y^\ast] &= \frac{1}{2\pi} \int_{-\pi}^\pi  |S(\eu^{j\omega})H(\eu^{j\omega)})|^2 F_N(\omega_\ell - \omega) \diff \omega.
    \end{aligned}
    \]
    The same steps are used to derive the variance of $\bar{V}_u$ and the cross-covariance between $\bar{V}_y$ and $\bar{V}_u$.
\end{proof}

We now state the main lemma.
\begin{lemma}\label{lemma:main}
Suppose that Assumption \ref{assumption} holds. Then, for any finite $N$, and in Case 1,
\[
\begin{aligned}
\sigma^{-1} (\Gdir - G ) &\leadsto \frac{1}{SR}(\bar{V}_y - G\bar{V}_u) \\
\sigma^{-1} (\Gind - G ) &\leadsto \frac{1}{S^2R}\bar{V_y}
\end{aligned}
\]
as $\sigma\to 0$, where $\bar{V}_y$ and $\bar{V}_u$ are zero mean statistically dependent complex random variables with finite variances $\bar{\sigma}_y$, $\bar{\sigma}_u$ and a cross-variance $\bar{\sigma}_{yu}$ as shown in Proposition \ref{prop:finite_dft_statistics}. On the other hand, in Case 2,
\[
\begin{aligned}
\sigma^{-1} (\Giodh - G ) &\leadsto \frac{1}{SR}(\bar{V}_y^{(1)} - G\bar{V}_u^{(2)})\\
\end{aligned}
\]
in which $\bar{V}_y^{(1)}$ is independent of $\bar{V}_u^{(2)}$.
\end{lemma}
\begin{proof}
    First observe that the estimators $\widehat{G}_{\bigcdot}$ are explicit smooth functions of $\sigma$. A first order Taylor expansion around $\sigma = 0$ can be obtained as
    \[
    \widehat{G}_{\bigcdot} = G + \deriv{\widehat{G}_{\bigcdot}}{\sigma}\biggr|_{\sigma = 0} \, \sigma + o(\sigma), \quad \sigma \to 0.
    \]
    Subtracting $G$ from both sides,  dividing by $\sigma$ and noticing that, by definition, ${o(\sigma)}/{\sigma}\to 0$  as $\sigma \to 0$, we see that  $\sigma^{-1}(\widehat{G}_{\bigcdot} - G) \leadsto \deriv{\widehat{G}_{\bigcdot}}{\sigma}\big|_{\sigma = 0}$ as $\sigma \to 0$. The results then follow from \eqref{eq:dir} and \eqref{eq:ind} for Case 1, and from \eqref{eq:dir2} for Case 2
\end{proof}

The result in Lemma \ref{lemma:main} is not a central limit theorem; it is however a small noise finite-sample result, and the limiting distribution as $\sigma \to 0$ is not necessarily complex Gaussian. The convergence takes place due to the vanishing noise, instead of increasing signal power as is usually the case in the available literature that relies on asymptotics in the data record length. Notice however that the limiting distributions of Lemma \ref{lemma:main} have finite variances that can be evaluated exactly and compared without making further assumptions on the noise. We refer to these variances as the asymptotic variances as $\sigma \to 0$. In the first case, they  are given as
\[
\begin{aligned}
  \asymvar[\Gdir]  & = \frac{1}{|SR|^2}(\bar{\sigma}_y + |G|^2\bar{\sigma}_u - 2 \Re[G^\ast\bar{\sigma}_{yu}]),\\
   \asymvar[\Gind]  & =\frac{1}{|S|^2|SR|^2}\bar{\sigma}_{y},
\end{aligned}
\]
and in the second case 
\[
\begin{aligned}
      \asymvar[\Giodh]  & = \frac{1}{|SR|^2}(\bar{\sigma}_y + |G|^2\bar{\sigma}_u).
\end{aligned}
\]

It is immediately clear from these expressions that in the frequency bands where $|S| \ll 1$ or $ |CG| \gg 1$ (i.e., where the loop shape has a high magnitude, and the closed loop control performance is good), the accuracy of the estimators in the two cases will be poor -- a known result in closed-loop identification. However, in other frequency bands, one of the two cases may have a better performance. \smallskip

Using the above asymptotic variance expressions, we arrive at the following theorem.

\begin{theorem}\label{thm:variance_inequality}
Suppose Assumption \ref{assumption} holds. Then, as $\sigma \to 0$,
\[
\asymvar[\Giodh] < \asymvar[\Gind]
\]
over all frequencies $\omega$ where $\Re[G^\ast\bar{\sigma}_{yu}] < 0$
\end{theorem}
\begin{proof}
    By direct comparison of the variance expressions.
\end{proof}

This shows that the correlation between the numerator and denominator of the defining fractions may improve or worsen the accuracy.
The definition of $\bar{\sigma}_{yu}$ in \eqref{eq:variances} is given in terms of a convolution, making the interpretation of the result in terms of the loop shape less evident. Fortunately, ignoring the leakage effects clarifies the result.

\begin{corollary}\label{cor:no_leackage}
    Suppose Assumption \ref{assumption} holds. Then, as $\sigma \to 0$, and neglecting leakage effects in the DFT of the noise terms,
    \[
      \asymvar[\Gdir] \approx \asymvar[\Gind] \approx   \frac{|H|^2}{|SR|^2},
    \]
    and
    \[
      \asymvar[\Giodh] \approx \frac{(1+|CG|^2)|H|^2}{|R|^2}.
    \]
Furthermore, the condition in Theorem \ref{thm:variance_inequality} is approximately equivalent to $\Re[CG]<0$.
\end{corollary}

\begin{proof}
Neglecting transient effects is equivalent to replacing the Fej\'{e}r kernel $F_N$ in the convolutions in \eqref{eq:variances} by a Dirac delta, resulting in
    \[
        \bar{\sigma}_y  \approx |SH|^2, \quad \bar{\sigma}_u    \approx |SCH|^2, \quad
        \bar{\sigma}_{yu} \approx |SH|^2C^\ast
    \]
The second part is evident by definition.
\end{proof}

\subsection*{Averaging over multiple experiments}
In cases where multiple independent data sets are available, the usual practice is to use averaging to improve the estimator accuracy. While arithmetic averages are possible, they still suffer from the infinite variance problem. Fortunately, geometrical averages over two or more experiments have a finite variance in general (except at $\omega =0$) \cite{Guillaume1996}. For the direct method, with two experiments,  the geometrically averaged estimator is
\[
  \overline{\Gdir} = \left( \frac{Y^{(1)}}{U^{(1)}} \times   \frac{Y^{(2)}}{U^{(2)}}\right)^{\frac{1}{2}},
\]
and the fact that $\asymvar\big[\,{\overline{\Gdir}}\,\big] = \frac{1}{2} \asymvar[\Gdir]$ follows in a straightforward manner by employing the same arguments used in Lemma \ref{lemma:main}. A geometrically averaged version of $\Giodh$ can similarly be defined as
\[
\overline{\Giodh} = \left(\frac{Y^{(1)}}{U^{(3)}} \times \frac{Y^{(2)}}{U^{(4)}}\right)^{\frac{1}{2}};
\]
its non-asymptotic variance is finite, and its asymptotic variance is half of that of $\Giodh$. 
\medskip

\section{Numerical Simulations}\label{sec:simulations}
To demonstrate the results, we performed a numerical simulation study in \textsc{Matlab} using the following transfer functions\footnote{this is the same system used in \cite{Vandenhof1993indirect}}
\[
\begin{aligned}
    G(\q) &= \frac{1}{1 - 1.6 \q^{-1} +0.89 \q^{-2}},\\[0.5em]
    C(\q) &= \q^{-1} - 0.8 q^{-2},\\[0.5em]
    H(\q) &= \frac{1-1.56 \q^{-1} + 1.045 \q^{-2} - 0.3338 \q^{-3}}{1-2.35 \q^{-1}+ 2.09 \q^{-2} -0.6675 \q^{-3}},
\end{aligned}
\]
connected as shown in Figure \ref{figure:cl},
so that the closed-loop system is stable and the noise $v$ is stationary.
\smallskip

 An experiment consisted of setting $r_1(k) = 0$ for all $k$, and using  $r_2$ as the excitation input. Multiple periods of a psuedo binary random signal with period  $N =127$ samples are used to excite the system    . Only the last period is used for identification, so that Assumption \ref{assumption}.(iii) holds. The  innovations  used a standard deviation $\sigma = 0.1$.
\smallskip

We ran $10^4$ independent MC simulations; each simulation consists of two independent experiments, which are used to evaluate the geometrically averaged estimators $\overline{\Gdir}$ and $\overline{\Giodh}$. The latter, in addition, used two independent copies of the input. The MC sample variances of the two estimators are then computed and compared with the asymptotic variances $\frac{1}{2}\asymvar[\Gdir]$ and $\frac{1}{2}\asymvar[\Giodh]$ (computed as in Corollary \ref{cor:no_leackage}). The results are shown in Figure \ref{fig:variances}, together with $|G|, |S|$, and the noise spectrum. 

Our first observation is that the  (non-asymptotic) sample variances follow closely the asymptotic variances  over the entire frequency range except between 0.4 and 0.7 rad/sec. The variances are large in this range where the sensitivity function is small. Figure \ref{fig:errors} shows the absolute value of the difference between the MC sample variances and the asymptotic variances. As predicted by Lemma \ref{lemma:main}, the  asymptotic variances give good approximations when the noise is relatively  small; the difference becomes larger at the peak of the noise spectrum.

Our second observation is that at low frequencies, the variance of $\Giodh$ is smaller than that of $\Gdir$, and the opposite occurs at high frequencies. This is exactly the behaviour predicted by our analysis; indeed as Figure \ref{fig:realGC} shows, the correlation between the numerator and denominator of the defining fraction of $\Gdir$ becomes negative only at frequencies $\omega > 0.65$ when $\Re[CG]$ is negative.

\begin{figure}
    \centering
    \includegraphics[width=\linewidth]{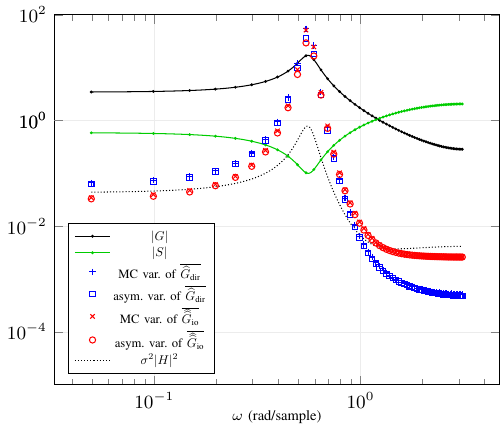}
     \caption{Variance comparison, together with $|G|$, $|S|$ and the noise spectrum. Sample variances are computed using $10^4$ MC simulations.}
    \label{fig:variances}
\end{figure}

\begin{figure}
    \centering
    \includegraphics[width=0.9\linewidth]{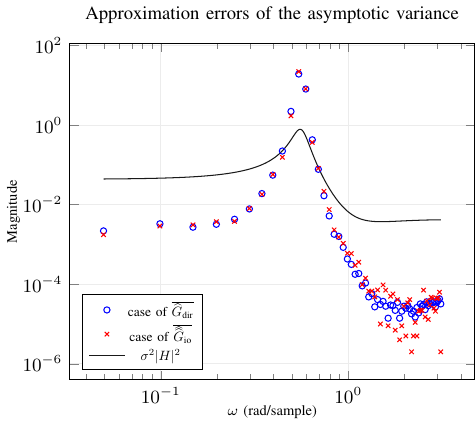}
     \caption{The absolute value of the difference between the actual (non-asymptotic) variances as approximated via Monte Carlo simulations, and the asymptotic variances. We also show the noise spectrum.   }
    \label{fig:errors}
\end{figure}

\begin{figure}
    \centering
    \includegraphics[width=0.9\linewidth]{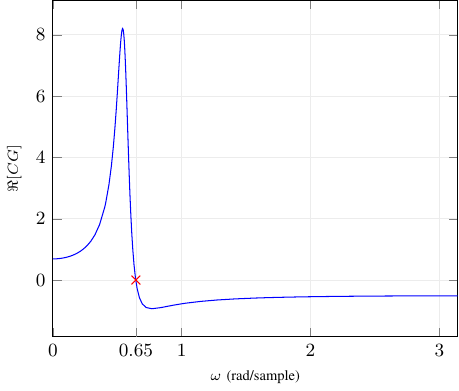}
     \caption{$\Re[GC]$ versus frequency: The sign of  $\Re[GC]$ switches to negative at $\omega = 0.64$ rad/sample. This is also the frequency around which the variance of $\overline{\Giodh}$ goes above that of $\overline{\Gdir}$.}
    \label{fig:realGC}
\end{figure}

\section{Conclusions}\label{sec:conclusions}

We provided a small noise analysis of the direct, indirect and joint input-output, non-parametric frequency domain closed-loop estimators. The results are given in terms of asymptotic variances (as $\sigma \to 0$) that are functions of the sensitivity function. As is well known, over the frequency ranges where the sensitivity function is small, $|S| \ll 1$, the closed-loop system suppresses the plant excitation leading to higher estimation variances for all methods.  

For larger values of the sensitivity, the performance of $\Gdir$ and $\Giodh$ depends on the loop shape. When $\Re[CG] <0$, the Nyquist plot is in the left half plane and corresponds to frequencies close to or larger than the cross over frequency. In these regions, the cross-correlation  $\bar{\sigma}_{yu}$ between $Y$ and $U$ is typically negative and the corresponding frequency-domain noises $V_y$ and $V_u$ have an obtuse/reflex angle between them. This means that their effects are in opposite directions, and the combined effect on the ratio $Y/U$ is, on average, larger than that of uncorrelated $V_y$ and $V_u$.

We also notice that for the open-loop case, where $|S| = 1, \; r = u $, the variance of the ETFE is $\sigma^2 |H|^2/ |U|^2$ and asymptotically in $\sigma$, the variance of the closed-loop direct estimator is approximately $\sigma^2 |H|^2/ |SR|^2$. But $U = SR$ in steady-state, and therefore, $\Gdir$ recovers the usual open-loop result for when $\sigma$ is small. On the other hand, the asymptotic variance of $\Giodh$  scales the open-loop variance with a factor $|S|^2(1+ |CG|^2)$. This means that $\Giodh$ may have a higher variance compared to the open-loop case.

\bibliographystyle{IEEEtran}
\bibliography{literature.bib}
\end{document}